\documentclass[letterpaper, 12 pt, conference,onecolumn]{ieeeconf}  

\IEEEoverridecommandlockouts                              
\usepackage{epsfig} 
\usepackage{amsmath,amssymb,amsfonts,dsfont,mathrsfs}
\usepackage{color}
\usepackage{balance}
\usepackage{booktabs,tabularx}
\usepackage{multirow}
\usepackage{mathtools}
\usepackage{pdfrender}
\usepackage{trfsigns}

\usepackage{centernot}
\usepackage{color}
\usepackage[acronym]{glossaries}
\usepackage[ruled]{algorithm2e}
\usepackage{graphicx,subcaption}
\usepackage{tabu}
\usepackage{cite}
\usepackage[official]{eurosym}

\let\labelindent\relax
\usepackage{enumitem}

\usepackage[prependcaption,colorinlistoftodos]{todonotes}

\graphicspath{{figures/}}

\newcommand{\expected}{\mathbb{E}}
\newcommand{\R}{\mathbb{R}}

\newcommand{\Z}{\mathbb{Z}}
\newcommand{\N}{\mathbb{N}}

\newcommand{\mc}[1]{\mathcal{#1}}

\newcommand{\bsone}{\boldsymbol{1}}

\newcommand{\col}{\mathrm{col}}
\newcommand{\inner}[2]{\langle #1, #2 \rangle}

\newtheorem{definition}{Definition}
\newtheorem{proposition}{Proposition}
\newtheorem{lemma}{Lemma}

\newtheorem{assumption}{Assumption}
\newtheorem{standing}{Standing Assumption}

\DeclareMathSizes{10}{10}{6}{4}

\makeglossaries
\newacronym{GNEP}{GNEP}{Generalized Nash Equilibrium Problem}
\newacronym{LP}{LP}{linear program}
\newacronym{DeePC}{DeePC}{data-enabled predictive control}
\newacronym{LTI}{LTI}{linear time-invariant}
\newacronym{ABCD}{ABCD}{alternate block-coordinate descent}

\title{\LARGE \bf
	The optimal transport paradigm enables data compression in data-driven robust control
}

\author{Filippo Fabiani and Paul J. Goulart 
	\thanks{The authors are with the Department of Engineering Science, University of Oxford, OX1 3PJ, United Kingdom {\tt \footnotesize (\{filippo.fabiani, paul.goulart\}@eng.ox.ac.uk)}. This work was partially supported through the Government’s modern industrial strategy by Innovate UK, part of UK Research and Innovation, under Project LEO (Ref. 104781).}%
}

\begin{document}

	\maketitle
	\thispagestyle{empty}
	\pagestyle{empty}

	\begin{abstract}

	A new data-enabled control technique for uncertain linear time-invariant systems, recently conceived by Coulson et\ al., builds upon the direct optimization of controllers over input/output pairs drawn from a large dataset. We adopt an optimal transport-based method for compressing such large dataset to a smaller synthetic dataset of representative behaviours, aiming to alleviate the computational burden of controllers to be implemented online. Specifically, the synthetic data are determined by minimizing the Wasserstein distance between atomic distributions supported on both the original dataset and the compressed one. We show that a distributionally robust control law computed using the compressed data enjoys the same type of performance guarantees as the original dataset, at the price of enlarging the ambiguity set by an easily computable and well-behaved quantity. Numerical simulations confirm that the control performance with the synthetic data is comparable to the one obtained with the original data, but with significantly less computation required.

	\end{abstract}

	\section{Introduction}
	In recent years, traditional model-based controller design has been giving way to data-driven approaches.
	The growing complexity of modern control problems, which often precludes the use of classical system identification procedures, along with the increasing data storage capacity, make \textit{learning from data} a timely and attractive new direction for the system-and-control community.

	Two recent trends in control originated from different, but nearly complementary, views of a well-known result in subspace identification, i.e., the so called \emph{Willems' lemma} \cite{willems2005note}. Specifically, this technical result establishes that a nonparametric \gls{LTI} realization of an unknown system can be recovered from a data matrix of noise-free input/output measurements. Thus, an early research direction adopted the Willems' lemma to perform explicit data-dependent parametrization analysis and control of systems \cite{de2019formulas,van2020data,berberich2020robust}. The second more recent direction, which is central in this paper, designed data-enabled (robust) optimal controllers without this intermediate modelling step \cite{coulson2019data,coulson2019regularized,berberich2020data}.

	To ensure that the column space of the data matrix spans all possible trajectories of a corresponding \gls{LTI} system, however, the Willems' lemma relies on the \emph{persistent excitation} of the system input. This directly translates to a requirement on the input/output observations length, and therefore one might be induced to collect extremely large datasets. As an undesired consequence, especially for the second line of research introduced above, this may pose several computational challenges to the real-time implementation of optimization-based control laws, thus limiting the scope of possible applications.
	Systems that require fast sampling rates, for instance, may not be amenable to this new approach to control design. Moreover, practitioners may choose short control horizons to alleviate the computational burden, thereby resulting in  poor control performance or even system instability. Designing suitable procedures to select the most representative data within large (possibly noisy) datasets is essential to overcome these problems.

	Data compression and dimensionality reduction have been widely adopted in the systems-and-control community to extrapolate synthetic sets of samples that ``best'' capture the information content of noise-corrupted datasets. Statistical and manifold learning \cite{roweis2000nonlinear, bartlett2006convexity}, principal component analysis \cite{wang2002new,wang2006closed,scholkopf1997kernel}, or subspace identification approaches \cite{van1997closed,mckelvey1996subspace,jansson1998consistency}, represent only a few of the most popular techniques. Conceptually, these methods tacitly neglect the stochastic nature of the noise itself, eliminating the less significant components (according to some criterion), or those potentially associated with noise, at a later stage only.

	In contrast to the aforementioned literature, we propose to design a synthetic set of samples via an offline procedure based on optimal transport \cite{villani2003topics,peyre2019computational} techniques. Specifically, the \emph{Wasserstein distance} represents the key tool to define an optimization problem that minimizes the distance between the discrete measure associated with our synthetic dataset and the empirical distribution of the original noisy data. The resulting \gls{LP} turns out to be a variational Wasserstein problem \cite[\S 9]{peyre2019computational}, wherein each atom of the synthetic dataset can be identified with a specific barycentre for a cluster of the original samples, thereby serving as a proxy a whole set of observed system behaviours. In the context of the data-enabled robust control technique of \cite{coulson2019regularized}, we show that a control law computed via the repeated solution of a distributionally robust optimization problem, built upon a Wasserstein ball as an ambiguity set, enjoys exactly the same type of performance guarantees as a controller designed using the full original dataset, at the price of enlarging the ambiguity radius by a quantity that depends on the number of synthetic atoms adopted.  We show that this additional ambiguity radius vanishes as the number of atoms in our synthetic dataset tends to the cardinality of the original dataset.

	The paper is organized as follows: we recall some fundamentals of optimal transport in \S II, and  introduce the robust control problem addressed in \S III. We formalize and discuss the variational Wasserstein problem associated with the design of synthetic datasets in \S IV. Finally, we compare the control performances obtained over the original and compressed datasets through numerical simulations in \S V.

	\smallskip
	\subsubsection*{Notation}
	For vectors $(x_1, \dots, x_N) \in\mathbb{R}^n$ and $\mc I \coloneqq \{1,\dots,N \}$, we denote $\mathrm{col}((x_i)_{i\in\mc I}) \coloneqq (x_1 ^\top,\dots ,x_N^\top )^\top$. Given a matrix $X \in \R^{n \times m}$, its $(i,j)$ entry is denoted by $[X]_{i,j}$. The symbol $\inner{\cdot}{\cdot}$ denotes an inner product in the appropriate space, i.e., $\inner{x}{y} = x^\top y$ for $(x, y) \in \R^n$ and $\inner{X}{Y} = \textrm{trace}\left(X^\top Y\right)$ for $(X, Y) \in \R^{n \times m}$. The probability simplex $\Sigma_n \coloneqq \{\sigma \in \R_+^n \mid \bsone_n^\top \sigma = 1 \}$, where $\bsone_n$ is $n$-dimensional vector of elements equal to $1$. For any point $x \in \R^n$, $\delta_x$ is the Dirac unit mass on $x$. Given a collection of points $\{x_i\}_{i \in \mc{I}} \in  \R^n$, $\textrm{conv}(\{x_i\}_{i \in \mc{I}})$ represents their convex hull, while $\hat{\mathds{P}} = \frac{1}{|\mathcal{I}|}\sum_{i \in \mc{I}}^{} \delta_{x_i}$ is the associated uniform discrete probability distribution. The dual norm of an arbitrary norm $\|\cdot\|$ on $\R^n$ is $\|x\|_\ast \coloneqq \textrm{sup}_{\|y\| \leq 1} \, \inner{x}{y}$. The conjugate function of $f : \R^n \to \R$ is defined by $f^\ast(\xi) \coloneqq \textrm{sup}_{x \in \R^n} \, \inner{\xi}{x} - f(x)$. Quantities denoted with $ (\hat{\cdot})$ are either directly measured, or depend on data.


	\section{Fundamentals of optimal transport}
	We first recall the definition of the Wasserstein distance for continuous measures. Then, by considering discrete distributions, we define the associated optimal transport problem that we will apply to our synthetic and empirical datasets.

	\subsection{Wasserstein distance between probability measures}
	Let $\Omega$ be an arbitrary space endowed with a metric $d$, and $\mc{P}(\Omega)$ be the set of Borel probability measures on $\Omega$.

	\smallskip
	\begin{definition}(\hspace{-.01em}\textup{\cite[Ch.~7]{villani2003topics}})\label{def:wass_dist}
		Given any $p \in [1, +\infty)$, the \emph{$p$-Wasserstein distance} $\mc{W}_p : \mc{P}(\Omega) \times \mc{P}(\Omega) \to \R_+$ between two probability measures $(\mathds{P}, \mathds{Q}) \in \mc{P}(\Omega)$ is defined as
		\begin{equation}\label{eq:wass}
		\mc{W}_p(\mathds{P}, \mathds{Q}) \coloneqq \left( \underset{\pi \in \Pi(\mathds{P}, \mathds{Q})}{\textrm{inf}} \, \int_{\Omega \times \Omega}^{} d^p(x,y) \, d \pi(x,y) \right)^{1/p},
		\end{equation}
		where $\Pi(\mathds{P}, \mathds{Q})$ denotes the set of all probability measures on $\Omega \times \Omega$ that have marginals $\mathds{P}$ and $\mathds{Q}$, respectively.
		\hfill$\square$
	\end{definition}
	\smallskip

	Roughly speaking, the decision variable $\pi$ of the infinite-dimensional optimization problem in \eqref{eq:wass} coincides with a transportation plan for moving a mass distribution described by $\mathds{P}$ to another one described by $\mathds{Q}$, while $d$ is the associated transportation cost. Then, given any $\varepsilon > 0$, we define the \emph{Wasserstein ball} of radius $\varepsilon$, centred around the distribution $\mathds{P}$, as $\mc{B}_{\varepsilon}(\mathds{P}) \coloneqq \{\mathds{Q} \in \mc{P}(\Omega) \mid \mc{W}_p(\mathds{P}, \mathds{Q}) \leq \varepsilon \}$.

	\subsection{Discrete probability distributions}
	Now, let us consider two families of $N$ and $M$ points in $\Omega$, i.e., $X = \{x_1, \ldots, x_N\}$ and $Y = \{y_1, \ldots, y_M\}$, respectively.
	Given weights $\alpha \in \Sigma_N$, $\beta \in \Sigma_M$, we can construct discrete probability distributions $\hat{\mathds{P}}$ and $\hat{\mathds{Q}}$ as $\hat{\mathds{P}} = \sum_{i \in \mc{N}}^{} \alpha_i \delta_{x_i}$ and $\hat{\mathds{Q}} = \sum_{i \in \mc{M}} \beta_i \delta_{y_i}$, where $\mc{N} \coloneqq \{1,\ldots,N\}$, $\mc{M} \coloneqq \{1,\ldots,M\}$. In this special case, the Wasserstein distance happens to correspond to the optimal value of a network problem, as \eqref{eq:wass} translates to the following \gls{LP} \cite{peyre2019computational}:

	\begin{equation}\label{eq:discrete_transp}
	\mc{W}_p(\hat{\mathds{P}}, \hat{\mathds{Q}}) = \underset{T \in \mc{T}(\alpha, \beta)}{\textrm{min}} \, \inner{T}{D(X,Y)}.
	\end{equation}

	Here, $D \in \R^{N \times M}$ is the matrix of pairwise distances between points in $X$ and $Y$, raised to the power $p$, defined as $[D]_{i,j} \coloneqq d^p(x_i, y_j)$, for all $x_i \in X$ and $y_j \in Y$. Moreover, every element of the decision matrix $T \in \R^{N \times M}$ in \eqref{eq:discrete_transp}, i.e., $[T]_{i,j} \eqqcolon t_{i,j}$, determines the coupling between pairs $(x_i,y_j) \in X \times Y$, whose value specifies the amount of mass flowing from the point $x_i \in X$ toward the point $y_j \in Y$. For any $\alpha \in \Sigma_N$ and $\beta \in \Sigma_M$, the admissible couplings lie in the feasible set $\mc{T} (\alpha, \beta)$, defined as
	\begin{equation}\label{eq:trans_poly}
	\mc{T} (\alpha, \beta) \coloneqq \{T \in \R_+^{N \times M} \mid T \, \bsone_M = \alpha, \, T^\top \, \bsone_N = \beta \}.
	\end{equation}

	The set of matrices in \eqref{eq:trans_poly} is called the \emph{transportation polytope}, since it is convex, bounded and defined by a set of $N + M$ equality constraints. Note that as long as $\mc{T} (\alpha, \beta)$ is nonempty, the solution to the \gls{LP} in \eqref{eq:discrete_transp}, attained on the vertices of $\mc{T}$, may not be unique. Finally, since $\mc{W}_p$ defines a metric, in the discrete setting \eqref{eq:discrete_transp} we have $\mc{W}_p(\hat{\mathds{P}}, \hat{\mathds{Q}}) = 0$ if and only if $\alpha = \beta$ \cite[Prop.~2.2]{peyre2019computational}, and therefore the triangle inequality holds as $\mc{W}_p(\hat{\mathds{P}}, \hat{\mathds{Q}}) \leq \mc{W}_p(\hat{\mathds{P}}, \hat{\mathds{G}}) + \mc{W}_p(\hat{\mathds{G}}, \hat{\mathds{Q}})$, for any discrete measures $\hat{\mathds{P}}$, $\hat{\mathds{Q}}$, $\hat{\mathds{G}} \in \mc{P}(\Omega)$.

 We henceforward focus on the \emph{Kantorovich-Rubinstein distance} obtained by setting $p = 1$, and consequently write $\mathcal{W}(\mathbb{P},\mathbb{Q})$ without subscript.  We will assume that the metric $d$ is induced by an arbitrary norm $\|\cdot\|$ on $\R^n$.

	\section{The data-enabled control paradigm}
	We start by formalizing the optimal control problem we wish to consider. After a brief digression on nonparametric models for deterministic \gls{LTI} systems, we will recall from \cite{coulson2019regularized} a distributionally robust reformulation of the control problem.

	\subsection{Constrained optimal control of uncertain \gls{LTI} systems}
	We consider discrete time, stochastic systems in the form:

	\begin{equation}\label{eq:stoc_LTI}
	\left\{
	\begin{aligned}
	x(k+1) &= A x(k) + B u(k) + E \nu(k),\\
	y(k) &= C x(k) + D u(k) + F \nu(k),
	\end{aligned}
	\right.
	\end{equation}
	where $A \in \R^{n \times n}$, $B \in \R^{n \times m}$, $E \in \R^{n \times q}$, $C \in \R^{\ell \times n}$, $D \in \R^{\ell \times m}$ and $F \in \R^{\ell \times q}$. The state, control input, output and disturbance at time instant $k \in \Z$ are $x(k) \in \R^n$, $u(k) \in \R^m$, $y(k) \in \R^\ell$ and $\nu(k) \in \R^q$, respectively. The uncertainty $\nu(k)$ is drawn from an unknown probability distribution $\mathds{P}_{\nu}$, supported on $\Upsilon \subseteq \R^{q}$. We will assume throughout that the system matrices defining \eqref{eq:stoc_LTI} are unknown, and that we have access to input/output measurements only, i.e., $(\hat{u}(k), \hat{y}(k))$, $k \in \Z$. Note that, in view of the dynamics in \eqref{eq:stoc_LTI}, any output measurement $\hat{y}(k)$ can be affected by the realization of the stochastic disturbance $\nu(k)$, for any $k \in \Z$.

	A typical approach to steer the behaviour of \eqref{eq:stoc_LTI}, particularly in the presence of state or input constraints, is stochastic model predictive control \cite{mesbah2016stochastic}. To this end, we consider a finite horizon control problem over horizon length $K \in \N$, where we aim to design a constrained sequence of control inputs, i.e., $u \coloneqq \col(u(k), \ldots, u(k+K-1)) \in \mc{U}$, for some compact, convex set $\mc{U} \subseteq \R^{mK}$, while minimizing a predefined cost function $J : \R^{mK} \times \R^{\ell K} \to \R$. Specifically, in view of the uncertain nature of \eqref{eq:stoc_LTI}, the finite horizon control problem translates into the following stochastic program:

	\begin{equation}\label{eq:stoc_opt_prob}
	\underset{u \in \mc{U}}{\textrm{inf}} \,\, \expected_{\mathds{P}_{\nu}^K} [J(u, y)],
	\end{equation}
	where $\mathds{P}_{\nu}^K \coloneqq \mathds{P}_{\nu} \times \ldots \times \mathds{P}_{\nu}$ is the $K$-fold product distribution characterizing $\nu$ over the whole horizon $K$. Next, we formulate the same working assumptions as in \cite{coulson2019regularized}.
	\smallskip
	\begin{standing}
		The pair $\!(\!A,B )\!$ is controllable.
		\hfill$\square$
	\end{standing}

	\smallskip
	\begin{standing}\label{standing:cost_fun}
		For all $(u,y) \in \R^{m K} \times \R^{\ell K}$, $J(u,y)$ is a separable function, namely $J(u,y) \coloneqq J_1(u) + J_2(y)$, where $J_1 : \R^{mK} \to \R$, $J_2 : \R^{\ell K} \to \R$ are convex and continuous. In addition, $J_2$ is such that $\Xi \coloneqq \{\xi \in \R^{\ell K} \mid J^\ast_2(\xi) < \infty \} \subseteq \R^{\ell K}$ is a bounded set.
		\hfill$\square$
	\end{standing}

	\subsection{Nonparametric models for deterministic \gls{LTI} systems}
	Let us first consider a deterministic version of the system in \eqref{eq:stoc_LTI}, i.e., with $\nu(k) = 0$, for any $k \in \Z$. The Willems' fundamental lemma \cite[Th.~1]{willems2005note} provides the theoretical means to construct data-consistent, minimal, nonparametric models for unknown, deterministic \gls{LTI} systems \cite{coulson2019regularized,de2019formulas,berberich2020robust,van2020data}. Specifically, assume that we have available (noise-free) data from an experiment of length $N$, $(\{\hat{u}(i)\}_{i = 0}^{N-1}, \{\hat{y}(i)\}_{i = 0}^{N-1})$, for different time shifts. Without loss of generality, assume also that $k = 0$ corresponds to the instant of initial observation. This data can then be organized within a matrix $\mathscr{H}_{K} \coloneqq \col(\hat{\mathscr{U}}_{0,K,N-K+1}, \hat{\mathscr{Y}}_{0,K,N-K+1}) \in \R^{(m+\ell)K \times N-K+1}$, where
	\begin{equation}\label{eq:Hankel}
			\hat{\mathscr{U}}_{0,K,N-K+1} \coloneqq \left[\begin{array}{cccc}
			\hat{u}(0) & \hat{u}(1) & \cdots & \hat{u}(N \! - \!K)\\
			\hat{u}(1) & \hat{u}(2) & \cdots & \hat{u}(N \!-\! K \!+\! 1)\\
			\vdots & \vdots & \ddots & \vdots\\
			\hat{u}(K \!-\! 1) & \hat{u}(K) & \cdots & \hat{u}(N \!-\! 1)
		\end{array}
		\right]
	\end{equation}
	and $\hat{\mathscr{Y}}_{0,K,N-K+1} \in \R^{\ell K \times N-K+1}$ is defined similarly. The first subscript refers to the time index of the top-left entry of a given matrix, the second refers to the number of block-rows, and the third to the number of columns. Note that both $\hat{\mathscr{U}}_{0,K,N-K+1} \in \R^{m K \times N-K+1}$ and $\hat{\mathscr{Y}}_{0,K,N-K+1}$ have constant vector entries along the block anti-diagonals, and therefore $\mathscr{H}_{K}$ belongs to the class of block-Hankel matrices.

	The Willems' fundamental lemma restricts the class of input sequences over the horizon $K$ to the persistently exciting signals, as defined next.

	\smallskip
	\begin{definition}(\hspace{-.01em}\textup{\cite{willems2005note}}) \label{def:persistent}
		 A measured signal $\hat{z} \in \R^{w}$, observed over $N$ samples, is \emph{persistently exciting} of order $K$ if the corresponding Hankel matrix $\hat{\mathscr{Z}}_{0,K,N-K+1}$, defined equivalently to \eqref{eq:Hankel}, has full rank $wK$.
		\hfill$\square$
	\end{definition}
	\smallskip


	It follows that for a signal to be persistently exciting of order $K$, its length $N$ must satisfy $N \geq (w+1)K-1$. Then, by relying on Definition~\ref{def:persistent}, we restate the Willems' fundamental lemma as follows:

	\smallskip
	\begin{lemma}(\hspace{-.01em}\textup{\cite[Th.~1]{willems2005note}}) \label{lemma:willems}
%
		Let $\col({\{\hat{u}(i)\}_{i = 0}^{N-1}})$ be a persistently exciting control signal of order $n + K$. Then $\col(u, y)$ is a $K$-long input/output trajectory of the deterministic version of the system in \eqref{eq:stoc_LTI} if and only if  $\col(u, y) \in \textrm{Im}(\mathscr{H}_{K})$.
		\hfill$\square$
	\end{lemma}
	\smallskip

		%
	Lemma~\ref{lemma:willems} establishes that if $N$ is chosen large enough and signals are persistently exciting, then every realisable input/output trajectory of the deterministic system is a linear combination of collected input/output data.  In other words, for any $K$-long input/output trajectory $\col(u, y)$, there will always exist some $g \in \R^{N-K+1}$ such that $\col(u, y) = \mathscr{H}_{K} g$.


	By making use of Lemma~\ref{lemma:willems}, our goal is next to restate the deterministic version of the finite horizon control problem in \eqref{eq:stoc_opt_prob} by rearranging a measured $N$-long, input/output trajectory, $\col(\hat{u},\hat{y}) \coloneqq (\{\hat{u}(i)\}_{i = 0}^{N-1}, \{\hat{y}(i)\}_{i = 0}^{N-1})$. Specifically, as in \cite[\S III.A]{coulson2019regularized}, by starting from the current time $k \in \N$, we assume that the control input $\hat{u}$ is persistently exciting of order $K_i + K + n$, for some $K_i \in \N$. We then split $\mathscr{H}_{K_i+K}$ into block matrices $\hat{\mathscr{U}}_f$, $\hat{\mathscr{Y}}_f$, $\hat{\mathscr{U}}_b$ and $\hat{\mathscr{Y}}_b$. Here, $\hat{\mathscr{U}}_f \in \R^{mK \times N- (K_i + K)+1}$ and $\hat{\mathscr{Y}}_f \in \R^{\ell K \times N- (K_i + K)+1}$ consist of the last $K$-block rows of $\mathscr{H}_{K_i+K}$, corresponding to data matrices for the ``forward'' propagation (i.e., from $k \in \N$ onward) of control sequence and output prediction, respectively. Conversely, $\hat{\mathscr{U}}_b \in \R^{m K_i \times N- (K_i + K)+1}$ and $\hat{\mathscr{Y}}_b \in \R^{\ell K_i \times N- (K_i + K)+1}$, which correspond to the first $K_i$-block rows of $\mathscr{H}_{K_i+K}$, define the consistency constraints associated with less recent measurements (i.e., ``backward'' data), together with $\hat{u}_i \coloneqq \col(\hat{u}(k - K_i),\ldots,\hat{u}(k-1))$ and $\hat{y}_i \coloneqq  \col(\hat{y}(k-K_i),\ldots,\hat{y}(k-1))$. Thus, the deterministic version of \eqref{eq:stoc_opt_prob} follows directly from Lemma~\ref{lemma:willems} and reads as:

	\begin{equation}\label{eq:det_opt_prob}
	\left\{
	\begin{aligned}
	&\underset{g}{\textrm{min}} & & J(\hat{\mathscr{U}}_f g, \hat{{\mathscr{Y}}}_f g)\\
	&\textrm{ s.t. } & &
	\left[\begin{array}{c}
	\hat{{\mathscr{U}}}_b\\
	\hat{{\mathscr{Y}}}_b
	\end{array}
	\right] g = \left[\begin{array}{c}
	\hat{u}_i\\
	\hat{y}_i
	\end{array}
	\right],\\
	&&& \hat{{\mathscr{U}}}_f g \in \mc{U},
	\end{aligned}
	\right.
	\end{equation}
	where the decision variable $g$ belongs to $\R^{N - (K_i + K) + 1}$.

	\subsection{A distributionally robust data-enabled control problem}
	For uncertain systems, the optimization problem in \eqref{eq:det_opt_prob} is complicated by the realization of noise terms $\nu(k)$
	drawn from the distribution $\mathds{P}_\nu$.  In particular, system noise complicates satisfaction of the consistency constraint $\hat{{\mathscr{Y}}}_b g = \hat{y}_i$.

	As proposed in \cite{coulson2019regularized}, a possible approach is to soften this consistency constraint, directly penalizing the term $\hat{{\mathscr{Y}}}_b g - \hat{y}_i$ in the cost function as follows
	\begin{equation}\label{eq:soften}
		\underset{g \in \mc{G}}{\textrm{min}} \; J(\hat{\mathscr{U}}_f g, \hat{{\mathscr{Y}}}_f g) + \rho \|\hat{\mathscr{Y}}_b g - \hat{y}_i\|_1,
	\end{equation}
	where $\mc{G} \coloneqq \{g \in \R^{N-(K_i + K)+1} \mid \hat{\mathscr{U}}_f g \in \mc{U}, \hat{\mathscr{U}}_b g = 	\hat{u}_i \}$ depends on input measurements only. The optimization problem in \eqref{eq:soften} can be manipulated to obtain a distributionally robust, semi-infinite reformulation. Specifically, we note that all random objects can be gathered into a matrix $\left[ \begin{smallmatrix}
		\mathscr{Y}_b & y_i\\
		\mathscr{U}_f & 0
	\end{smallmatrix} \right]$, whose $j$-th row is denoted by $\kappa^\top_j$. Any such row corresponds to a random vector distributed according to some probability $\mathds{P}_{\kappa_j}$, and supported on $\Theta_{k_j} \subseteq \R^{N - (K_i + K) + 2}$, for all $j \in \{1, \ldots, \ell(K_i + K)\}$. Note that every $\mathds{P}_{\kappa_j}$ and $\Theta_{k_j}$ is determined starting from the unknown distribution $\mathds{P}_\nu$ and support $\Upsilon$, respectively.
Define $\kappa \coloneqq \col((\kappa_j)_{j = 1}^{\ell(K_i + K)})$, a random vector supported on $\Theta \coloneqq \prod_{j = 1}^{\ell(K_i + K)} \Theta_{k_j} \subseteq \R^{\ell(K_i + K)(N - (K_i+K) + 2)}$ and distributed according to $\mathds{P}_\kappa \coloneqq \prod_{j = 1}^{\ell(K_i + K)} \mathds{P}_{\kappa_j}$, and let $v \coloneqq \col(g, -1)$ and $\mc{V} \coloneqq \mc{G} \times \{-1\} \subseteq \R^{N-(K_i + K)+2}$. With the  notation introduced, and taking into account the realization of the random objects in $\kappa$, the cost function in \eqref{eq:soften} turns out to be $J((\hat{\mathscr{U}}_f, 0) v, (\hat{\kappa}_{\ell K_i + 1} v, \ldots, \hat{\kappa}_{\ell (K_i + K)} v)) + \rho \| (\hat{\kappa}_1 v, \ldots, \hat{\kappa}_{\ell K_i} v) \|_1 \eqqcolon f(\hat{\kappa},v)$. Therefore, based on the measurements $\hat{\kappa}$, the so-called \emph{in-sample performance} of \eqref{eq:soften} is
	$
		\textrm{min}_{v \in \mc{V}} \; \expected_{\hat{\mathds{P}}_\kappa} [f(\kappa,v)],
	$
	and admits the following distributionally robust, semi-infinite variation over the Wasserstein ball centred at the empirical distribution $\hat{\mathds{P}}_\kappa$:
	\begin{equation}\label{eq:distrib_robust}
		\underset{v \in \mc{V}}{\textrm{inf}} \; \underset{\mathds{Q} \in \mc{B}_{\varepsilon}(\hat{\mathds{P}}_\kappa)}{\textrm{sup}} \, \expected_{\mathds{Q}} [f(\kappa,v)].
	\end{equation}

	The optimal value of \eqref{eq:distrib_robust} is known to upper bound the \emph{out-of-sample performance}, $\expected_{\mathds{P}_\kappa} [f(\kappa, v)]$, with high confidence \cite{esfahani2018data,coulson2019regularized}. Note that $\expected_{\mathds{P}_\kappa} [f(\kappa, v)]$ denotes the quantity of interest in studying \eqref{eq:stoc_opt_prob}, as it depends on the unknown distribution $\mathds{P}_\kappa$. We will assume that this distribution is light-tailed, which is key to the results in \cite{esfahani2018data,coulson2019regularized}:
	\smallskip
	\begin{assumption}\label{ass:light_tailed}
		There exists some $a > 0$ such that $\expected_{\mathds{P}_{\kappa}} [e^{{\|\kappa\|^a}}] \coloneqq \int_{\Theta}^{} e^{{\|\kappa\|^a}} \, \mathds{P}_{\kappa}(d \kappa) < \infty$.
		\hfill$\square$
	\end{assumption}
	\smallskip

 Under Assumption~\ref{ass:light_tailed}, which is satisfied automatically if $\Theta$ is compact, \cite[Th.~3.5]{esfahani2018data} establishes that for any given confidence parameter $\beta > 0$, there exists some data-driven ambiguity radius, $\varepsilon = \varepsilon(\beta) > 0$, which guarantees the following probabilistic bound
	\begin{equation}\label{eq:performance_orig}
			\mathds{P}^{K_i + K}_{\kappa}\left\{ \expected_{\mathds{P}_{\kappa}} [f(\kappa,v)] \leq \underset{\mathds{Q} \in \mc{B}_{\varepsilon}(\hat{\mathds{P}}_\kappa)}{\textrm{sup}} \ \expected_{\mathds{Q}} [f(\kappa,v)] \right\}	\geq 1 - \beta.
	\end{equation}

	However, solving the problem \eqref{eq:distrib_robust} is not trivial to solve since it is semi-infinite. Using the results in \cite{esfahani2018data}, \cite[Th.~4.2]{coulson2019regularized} shows that considering $\mc{B}_{\varepsilon}(\hat{\mathds{P}}_\kappa)$ as an ambiguity set in \eqref{eq:distrib_robust} allows for a finite, convex reformulation.
	Specifically, the optimal value of \eqref{eq:distrib_robust} is upper bounded by
	\begin{equation}\label{eq:opt_reformulation}
	\underset{v \in \mc{V}}{\textrm{min}} \; \left[f(\hat{\kappa}, v) + \varepsilon \cdot \textrm{max} \left( \underset{\xi \in \Xi}{\textrm{sup}} \; \|\xi\|_\infty \| \col(g,0) \|_\ast, \rho \|v\|_\ast \right)\right].
	\end{equation}

	Thus, by denoting $v^\star$ as an optimal solution to \eqref{eq:opt_reformulation}, the control law $u^\star = \hat{\mathscr{U}}_f g^\star$ enjoys the data-driven probabilistic guarantees in \eqref{eq:performance_orig}, obtaining good control performance with respect to the possible realizations of the stochastic output trajectory $y$ associated with the ambiguity set $\mc{B}_{\varepsilon}(\hat{\mathds{P}}_\kappa)$.

	Some consideration of the optimization problem in \eqref{eq:opt_reformulation} is in order. First, we note that the cost function is convex since it corresponds to the sum of a (separable) convex function, $f$, 
	and the pointwise maximum between dual norms, and is therefore also convex. Moreover, by defining $\zeta(\cdot) \coloneqq \rho \, \|\cdot\|_1$, $ \textrm{max} \left( \textrm{sup}_{\xi \in \Xi} \; \|\xi\|_\infty \| \col(g,0) \|_\ast, \rho \|v\|_\ast \right)$ is equivalent to
	$$
		\left\{
		\begin{aligned}
		&\underset{\lambda \geq 0}{\textrm{inf}} & & \lambda \varepsilon\\
		&\textrm{ s.t. } & &	\textrm{sup}_{\xi \in \Xi} \; \|\xi\|_{\infty} \|\col(g,0)\|_{\ast} \leq \lambda,\\
		&&& \textrm{sup}_{\xi \in \Xi'} \; \|\xi\|_{\infty} \|v\|_{\ast} \leq \lambda,
		\end{aligned}
		\right.
	$$
	where $\Xi' \coloneqq \{\xi \in \R^{\ell K} \mid \zeta^\ast(\xi) < \infty \}$ is a bounded set. In fact, by the definition of the conjugate function, $\zeta^\ast(\xi) = 0$ if $\|\xi\|_{\infty} \leq \rho$, while $\zeta^\ast(\xi) = \infty$ otherwise.  It therefore follows that \eqref{eq:opt_reformulation} amounts to solving a conic optimization problem where the dimension of the decision variable $v$ is $N - (K_i + K) + 1$, and $N$ is a design parameter chosen so that $N \geq (m + 1)(n+ K_i + K) -1$, i.e., the persistent excitation assumption is satisfied.
	The large amount of data that one should collect may pose several challenges in the online implementation of distributionally robust controllers. Alternatively, one may also choose short control horizons $K$, resulting in poor control performance or instability of the controlled system. This motivates us to develop a procedure to compress the information brought by the dataset $\mathscr{H}_{K_i+K}$ into a smaller synthetic set of representative system behaviours. In the next section, we formalize the data compression problem as a variational Wasserstein problem, proposing a solution procedure to design such a synthetic  dataset.

	\section{Data compression as a variational Wasserstein problem}

	In this section we propose an offline, optimal transport-based procedure that can compress a possibly large dataset of system trajectories to a smaller, synthetic one of representative behaviours, i.e., to select a limited set of synthetic representative samples that ``best'' summarize the information content.
	This clearly affects the control problem addressed by imposing a lower computational burden in solving an optimization problem similar to \eqref{eq:opt_reformulation}, hence making distributionally robust control approaches more appealing for online implementation.
	Moreover, we show that the optimal solution obtained by means of our synthetic dataset enjoys probabilistic guarantees of the same type in \eqref{eq:performance_orig} on a Wasserstein ball with an enlarged ambiguity radius. Essentially, the marginal increase in the radius of the Wasserstein ball depends on the number of samples adopted, and vanishes as the size of the synthetic dataset grows to that of the original dataset.

	Let the dimensions of a matrix of input/output measurements over some horizon $L \in \N$, $\mathscr{H}_L \in \R^{r \times R}$, be fixed, i.e., $N \geq (m+1) (n + L) - 1$ be chosen so that the control input $\hat{u}$ is persistently exciting of order $n + L$ as in Lemma~\ref{lemma:willems}, where $r \coloneqq (m + \ell)L$, $R \coloneqq N - L + 1$. The empirical distribution of such measurements is defined as
	$
		\hat{\mathds{P}}_\kappa = \tfrac{1}{R} \sum_{i \in \mc{R}} \delta_{h_i},
	$
	where $\mc{R} \coloneqq \{1,\ldots,R\}$ and $h_i$ is the $i$-th column of $\mathscr{H}_L$. Our goal is to find a set of locations $\mathscr{S}_L  := \left[s_1 \dots s_S \right]\in \R^{p \times S}$, with $S \leq R$, whose empirical probability distribution $\hat{\mathds{P}}_{s} = \tfrac{1}{S} \sum_{i \in \mc{S}} \delta_{s_i}$, $\mc{S} \coloneqq \{1,\ldots,S\}$, is closest to that of the original dataset. Hence, our problem can be formulated as an optimal transport problem
	\begin{equation}\label{eq:wass_optimal}
		\underset{\mathscr{S}_L}{\textrm{min}} \; \mc{W}(\hat{\mathds{P}}_\kappa, \hat{\mathds{P}}_{s}) = \underset{\mathscr{S}_L}{\textrm{min}} \; \underset{T \in \mc{T}(\bsone_R/R, \bsone_S/S)}{\textrm{min}} \inner{T}{D(\mathscr{H}_L, \mathscr{S}_L)}.
	\end{equation}

 Our optimization problem has a strong practical interpretation. Specifically, an optimal solution to \eqref{eq:wass_optimal}, $\mathscr{S}_L^\star$, is one that produces a distribution $\hat{\mathds{P}}_{s}$ closest to the original one in the Wasserstein distance, and is therefore the one that minimizes the transport cost between the two distributions $\hat{\mathds{P}}_s$ and $\hat{\mathds{P}}_\kappa$. With a slight abuse of notation, we define by $\eta(S)$ the optimal value to \eqref{eq:wass_optimal}, which clearly depends on $S$, the number of samples adopted. This quantity is key in characterizing the robustness properties of the control approach that we are about to introduce.
	The nested optimization program in \eqref{eq:wass_optimal} is a variational Wasserstein problem, representing a particular case of the Wasserstein barycentres problem \cite{cuturi2014fast}. Specifically, it directly falls into the set of $k$-means problems \cite{ng2000note}. In addition, for semi-discrete settings the benefit of adopting the Wasserstein distance as a metric to compare probability measures has been proved in many theoretical and practical problems, from dictionary and statistical learning \cite{rolet2016fast,frogner2015learning}, to vision and image processing \cite{cuturi2016smoothed,lellmann2014imaging}.

	\subsection{On the optimal transport problem \eqref{eq:wass_optimal}}

	\begin{figure}[t!]
		\centering
		\includegraphics[width=.5\columnwidth]{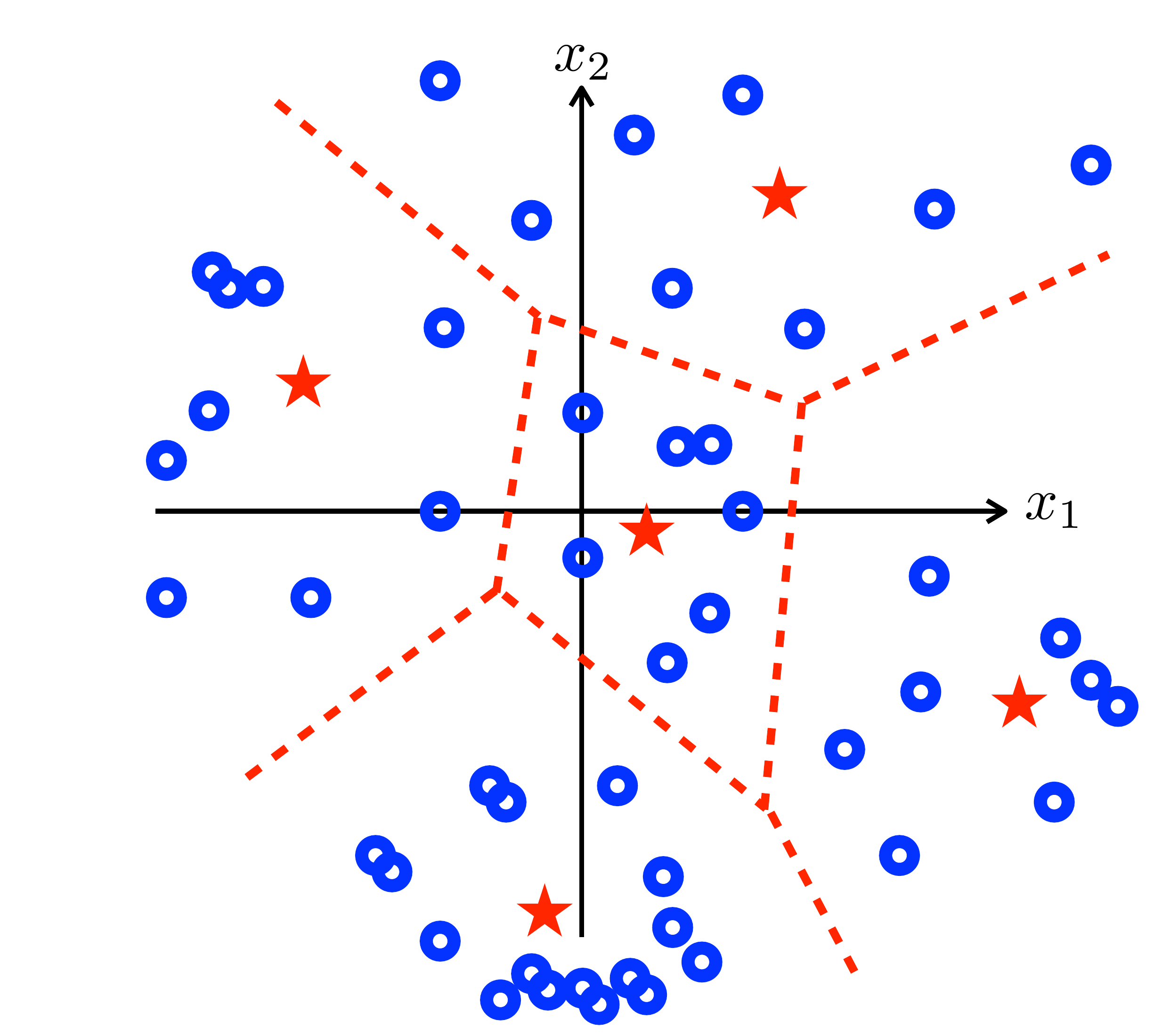}
		\caption{Schematic two-dimensional representation of the Wasserstein barycentres problem in \eqref{eq:barycentres}. Each synthetic atom (red star) identifies a specific barycentre associated to a cluster (dashed red lines) of original samples (blue circles). }
		\label{fig:opt-transport}
	\end{figure}

	Despite its appealing structure and strong practical interpretation, the variational Wasserstein problem in \eqref{eq:wass_optimal} is convex in each single variable, i.e., $\mathscr{S}_L$ and $T$, but not jointly. In practice, this might produce locally optimal solutions, where every atom defining each optimal $\mathscr{S}^\star_L$ identifies a specific barycentre for a subset of samples in $\mathscr{H}_L$. Specifically, in our setting, \eqref{eq:wass_optimal} corresponds to
	\begin{equation}\label{eq:barycentres}
			\underset{\mathscr{S}_L}{\textrm{min}} \; \underset{T \in \mc{T}(\bsone_R/R, \bsone_S/S)}{\textrm{min}} \; \sum_{j \in \mc{S}} \, \sum_{i \in \mc{R}} \, t_{i,j} \|h_i - s_j\|.
	\end{equation}

	For any fixed $j \in \mc{S}$, every $t_{i,j} \geq 0$ defines the quantity of a predefined sample $h_i$ of $\mathscr{H}_L$ that is associated with a barycentre $s_j$ -- see, e.g., Fig.\ref{fig:opt-transport} for a schematic representation. In this specific example, each sample $h_i$ is associated to the $j$-th barycentre only, i.e., $t_{i,j} = 1$ and $t_{i,h} = 0$, for all $h \in \mc{S}\setminus\{j\}$, meaning that there are no overlapping clusters. In principle, as $t_{i,j} \geq 0$, the information associated with any sample $h_i$ may be distributed among several barycentres, thereby producing overlapping clusters. It is easily shown that every atom defining $\mathscr{S}^\star_L$ belongs to $\textrm{conv}(\mathscr{H}_L)$.

	Since the inner minimization problem in \eqref{eq:barycentres} represents the pointwise minimum of linear functions, we note that the Wasserstein distance is not smooth in its arguments.
	To circumvent this problem, the cost function in \eqref{eq:wass_optimal} can be regularized by means of a strictly convex, weighted entropic term, i.e., $\gamma \, \inner{T}{\log(T)}$, for some $\gamma > 0$. The benefits are twofold \cite{cuturi2013sinkhorn}: i) the inner optimization problem in \eqref{eq:wass_optimal} admits a closed form, which translates to a matrix balancing problem, and ii) the Wasserstein distance is differentiable.
	To solve \eqref{eq:barycentres} one may also rely on the convexity in each single variable of the Wasserstein distance, for which possible solution algorithms are, e.g., the typical alternate block-coordinate descent methods \cite{bertsekas1989parallel,beck2013convergence}.

%
%

	\subsection{Robust performance guarantees}
	From \cite[Th.~3.5]{esfahani2018data}, we know that, for any confidence parameter $\beta$, there exists some $\varepsilon > 0$, that depends only on the amount of available data, such that the optimal value of \eqref{eq:distrib_robust} upper bounds the out-of-sample $\expected_{\mathds{P}_\kappa} [f(\kappa, v)]$. Our approach then amounts to first solving the optimal transport problem in \eqref{eq:wass_optimal}, computing a synthetic set of atoms $\mathscr{S}_L$, and then reformulating the robust optimization problem in \eqref{eq:distrib_robust} with a Wasserstein ball centred on $\hat{\mathds{P}}_s$ instead of $\hat{\mathds{P}}_\kappa$, i.e.,
	\begin{equation}\label{eq:distr_robust_synthetic}
				\underset{v \in \mc{V}}{\textrm{inf}} \; \underset{\mathds{Q} \in \mc{B}_{\varepsilon}(\hat{\mathds{P}}_s)}{\textrm{sup}} \, \expected_{\mathds{Q}} [f(\kappa,v)].
	\end{equation}

	The fact that $S < R$ in general, namely we are designing a reduced set of synthetic samples relative to the original dataset, intuitively has two main implications:
	\begin{enumerate}
	\item[i)] The optimal value in \eqref{eq:distr_robust_synthetic} can still achieve the performance bound in \eqref{eq:performance_orig} with high confidence, at a price of considering a larger radius of the ambiguity set $\mc{B}_{\varepsilon}(\hat{\mathds{P}}_s)$ (see Proposition~\ref{prop:main}). We show that the additional term accounting for additional robustness vanishes as the number of atoms tends to the cardinality of the original dataset;
	\item[ii)]  The robust optimization problem in \eqref{eq:distr_robust_synthetic} can be manipulated to obtain a tractable convex reformulation equivalent to the one in \eqref{eq:opt_reformulation}, but defined on a lower dimensional space. In fact, while the optimization variable in \eqref{eq:opt_reformulation} has dimension $R = N - L + 1$, where $N$ is chosen so that the condition on the persistency of excitation is met, which leads to $R \geq (n+1) + \tfrac{1}{m+1} (m N - 1)$, the finite, convex formulation obtained by manipulating \eqref{eq:distr_robust_synthetic} establishes that $v$ simply belongs to $\R^S$.
\end{enumerate}

	We remark that, instead, the number of constraints defining $\mc{V}$ does not change since it depends on the control horizon $L$ (design parameter). Despite this dimensionality mismatch, for simplicity's sake we will keep the same notation in the rest of the paper.

	Next, by relying on the definition of $\eta(S)$ following \eqref{eq:wass_optimal}, we show that an optimal solution to the robust optimization problem in \eqref{eq:distr_robust_synthetic} upper bounds the  out-of-sample performance $\expected_{\mathds{P}_{\kappa}} [f(\kappa,v)]$ with high confidence.

	\smallskip
	\begin{proposition}\label{prop:main}
		Let $\beta \in (0,1)$ be some given confidence parameter, and let $S \leq R$ be fixed. Under Assumption~\ref{ass:light_tailed}, there exists some $\bar{\varepsilon} = \bar{\varepsilon}(\beta, S) > 0$ such that, for all $v \in \mc{V}$,
		$$
		\mathds{P}^S_\kappa \left\{	\expected_{{\mathds{P}_\kappa}}[f(\kappa, v)] \leq \underset{\mathds{Q} \in \mc{B}_{\bar{\varepsilon}}(\hat{\mathds{P}}_s)}{\textrm{sup}}  \expected_{\mathds{Q}} [f(\kappa,v)]	\right\} \geq 1 - \beta.
		$$
		\hfill$\square$
	\end{proposition}
	\begin{proof}
		First, given any $S \leq R$, we recall that $\eta(S)$ corresponds to the Wasserstein distance between the discrete probability distribution associated with the original dataset, $\hat{\mathds{P}}_\kappa$, and the computed synthetic one, $\hat{\mathds{P}}_s$. Thus, the triangle inequality for the Wasserstein metric ensures that the distance between the real, unknown distribution $\mathds{P}_\kappa$ and $\hat{\mathds{P}}_s$ can be upper bounded as follows
		$$
		\mc{W}(\mathds{P}_\kappa, \hat{\mathds{P}}_s) \! \leq \! \mc{W}(\mathds{P}_\kappa, \hat{\mathds{P}}_\kappa) + \mc{W}(\hat{\mathds{P}}_\kappa, \hat{\mathds{P}}_s) \! = \! \mc{W}(\mathds{P}_\kappa, \hat{\mathds{P}}_\kappa) + \eta(S).
		$$
		Moreover,  in view of Assumption~\ref{ass:light_tailed}, it follows from \cite[Th.~3.4]{esfahani2018data} that, for any fixed $\beta \in (0,1)$, there exists some $\varepsilon(\beta) > 0$ such that $\mathds{P}^R_\kappa\{	\mc{W}(\mathds{P}_\kappa, \hat{\mathds{P}}_\kappa) \leq \varepsilon(\beta)	\} \geq 1 - \beta$. Therefore, since $\hat{\mathds{P}}_s$ is an empirical distribution as well, we obtain $\mathds{P}^S_\kappa\{	\mc{W}(\mathds{P}_\kappa, \hat{\mathds{P}}_s) \leq \varepsilon(\beta) + \eta(S)	\} \geq 1 - \beta$. This latter relation, which can be equivalently restated as $\mathds{P}^S_\kappa\{	\mathds{P}_\kappa \in \mc{B}_{\bar{\varepsilon}}(\hat{\mathds{P}}_s) 	\} \geq 1 - \beta$, where $\bar{\varepsilon} \coloneqq \varepsilon(\beta) + \eta(S)$, directly implies $\expected_{{\mathds{P}_\kappa}}[f(\kappa, v)] \leq {\textrm{sup}}_{\mathds{Q} \in \mc{B}_{\bar{\varepsilon}}(\hat{\mathds{P}}_s)}  \expected_{\mathds{Q}} [f(\kappa,v)]$ with probability $1 - \beta$, thus  concluding the proof.
	\end{proof}
	\smallskip

	We note that the Wasserstein distance between $\hat{\mathds{P}}_\kappa$ and the empirical distribution associated with the compressed dataset, $\hat{\mathds{P}}_s$, can be made arbitrarily small by increasing the number of atoms $S$, since $\textrm{min}_{\mathscr{S}_L} \mc{W}(\hat{\mathds{P}}_\kappa, \hat{\mathds{P}}_{s}) \to 0$ as $S \to R$, and hence $\eta(S) \to 0$. In this case, we recover the radius of the ambiguity set in \cite[Th.~4.1]{coulson2019regularized}, although the behaviour of $\eta(S)$ is not monotonically decreasing to $0$ as we will see in the next section.
	Finally, we remark that the optimization problem in \eqref{eq:distr_robust_synthetic} admits a tractable reformulation identical to the one in \eqref{eq:opt_reformulation}, but which can be solved online with significantly lower computational burden. This then paves the way to adopt possibly longer control horizons $L$, thus enhancing the control performance and without compromising the closed-loop stability of the system. These aspects are investigated in the next section.

	\section{Numerical simulations}
		\begin{algorithm}[!t]
		\caption{Receding horizon robust synDeePC}\label{alg:deepc}
		\DontPrintSemicolon
		\SetArgSty{}
		\SetKwFor{ForAll}{for all}{do}{end forall}
		\textbf{Offline:} Given $\mathscr{H}_{K_i + K}$, set $S \leq N - (K_i + K) + 1$, compute $\mathscr{S}^\star_{K_i + K} \in \textrm{argmin}_{\mathscr{S}_{K_i + K}} \; \mc{W}(\hat{\mathds{P}}_\kappa, \hat{\mathds{P}}_{s})$\\
		\smallskip
		\hrule width.93\columnwidth
		\smallskip
		\textbf{Initialization:} Set $\hat{\mc{V}}(0)$, $\hat{\kappa}(0)$ and $\bar{\varepsilon}$\\
		\smallskip
		\textbf{Iteration $(k \in \N)$:} \\
		\begin{itemize}\setlength{\itemindent}{.3cm}
			\smallskip
			\item[(\texttt{S1})] Compute
			$$
			v^\star(k) \coloneqq \underset{v \in \mc{V}	(k)}{\textrm{argmin}} \; \left[f(\hat{\kappa}(k), v) +  \bar{\varepsilon} \cdot \textrm{max} \left( \textrm{sup}_{\xi \in \Xi} \; \|\xi\|_\infty \| \col(g,0) \|_\ast, \rho \|v\|_\ast \right)\right]
			$$

			\smallskip

			\item[(\texttt{S2})] Set $u^\star(k) = \mathscr{U}_f v^\star(k)$, apply $u_1^\star(k)$

			\smallskip

			\item[(\texttt{S3})] Collect measurements, update $\mc{V}(k\!+\!1)$, $\hat{\kappa}(k\!+\!1)$
		\end{itemize}
	\end{algorithm}
	\begin{table}[!t]
		\caption{Main simulation parameters}\label{tab:sim_param}
		\begin{center}
			\begin{tabular}{ccccccccc}
				\toprule
				$K$ & $K_i$ & $N$  & $T_s$ & $c$	& $\rho$	& $\epsilon(\beta)$\\
				\midrule
				$30$ & $1$ & $214$ & $0.05$ & $200$  & $10^5$ & $10^{-3}$ \\
				\bottomrule
			\end{tabular}
		\end{center}
	\end{table}
	In this section we apply the \gls{DeePC} method presented in \cite{coulson2019regularized} with an additional, offline step that computes a synthetic dataset, thus renamed synDeePC (Algorithm~\ref{alg:deepc}). Specifically, we compare control and computational performance when considering the original dataset and a compressed one when steering a linear model of a quadcopter in a receding horizon fashion.

	Simulations are run in Matlab by using Gurobi \cite{gurobi} as a solver for (\texttt{S1}) in Algorithm~\ref{alg:deepc}, on a laptop with a Quad-Core Intel Core i5 2.4 GHz CPU and 8 Gb RAM. The main parameters adopted are summarized in Table~\ref{tab:sim_param}.

	The linear model adopted is valid around a hover position, where the state vector is $\col(x,y,z,\dot{x},\dot{y},\dot{z},\phi,\theta,\psi,\dot{\phi},\dot{\theta},\dot{\psi}) \in \R^{12}$. Here, $x$, $y$ and $z$ are the three spatial coordinates and relative velocities ($\dot{x},\dot{y},\dot{z}$), while $\phi$, $\theta$, and $\psi$ are the angular ones, with relative rates ($\dot{\phi},\dot{\theta},\dot{\psi}$). The control inputs are represented by four identical rotors, constrained to the set $\mc{U} = [-0.7007, 0.2993]^{4(K_i + K)}$ due to physical limitations. By assuming full state measurement, we use the same state-space matrices adopted in \cite[\S V]{coulson2019regularized}, as well as same original cost function, $J(u,y) = \|u\|_1 + c \|y - r\|_1$, where $r$ denotes a parametrized, 8-figure trajectory with fixed altitude. The parameter $T_s$ in Table~\ref{tab:sim_param} represents the temporal resolution with which the reference trajectory is sampled (i.e., the sampling time). Moreover, with the adopted values, the $184$ columns of the matrix $\mathscr{H}_{31}$ are filled by means of random inputs drawn from a uniform distribution on $\mc{U}$: this is to guarantee the persistency of excitation for the (syn)DeePC, according to Definition~\ref{def:persistent}.

	\begin{figure}
		\centering
		{\includegraphics[width=0.5\columnwidth,trim=0 0 0 1cm]{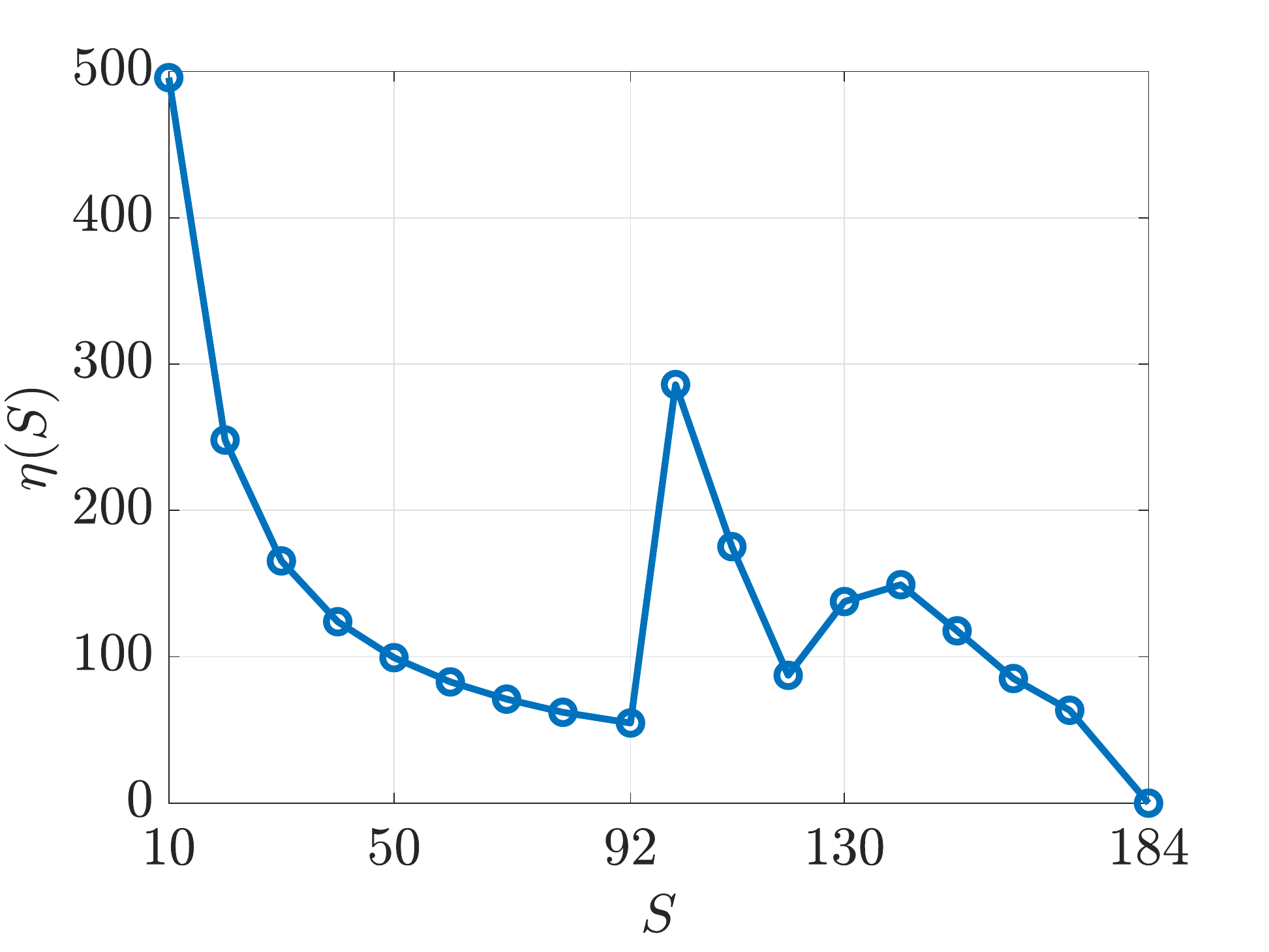}}
		\caption{Wasserstein distance vs.\ number of synthetic atoms.}
		\label{fig:cost_evol}
	\end{figure}

	Some a-priori considerations on the choice of the parameter $S$ in Algorithm~\ref{alg:deepc} can be made, e.g., in a data-driven fashion. Specifically, we evaluate the behaviour of the Wasserstein distance $\mc{W}(\hat{\mathds{P}}_\kappa, \hat{\mathds{P}}_s)$ when $S$ varies, solving the offline step without regularization by means of a standard block-descent algorithm \cite{bertsekas1989parallel}. Thus, according to Fig.~\ref{fig:cost_evol}, the function $\eta(S)$ takes reasonable values for $S \leq 92$, leading to an offline step in Algorithm~\ref{alg:deepc} taking less than three minutes. Interestingly, we note that the effect of the local minima seems to prevent the Wasserstein distance from decreasing monotonically for values of $S > 92$, i.e., $R/2$, as in Fig.~\ref{fig:cost_evol}.

	\begin{figure}
		\centering
		{\includegraphics[width=0.5\columnwidth,trim = 0 0 0 1cm]{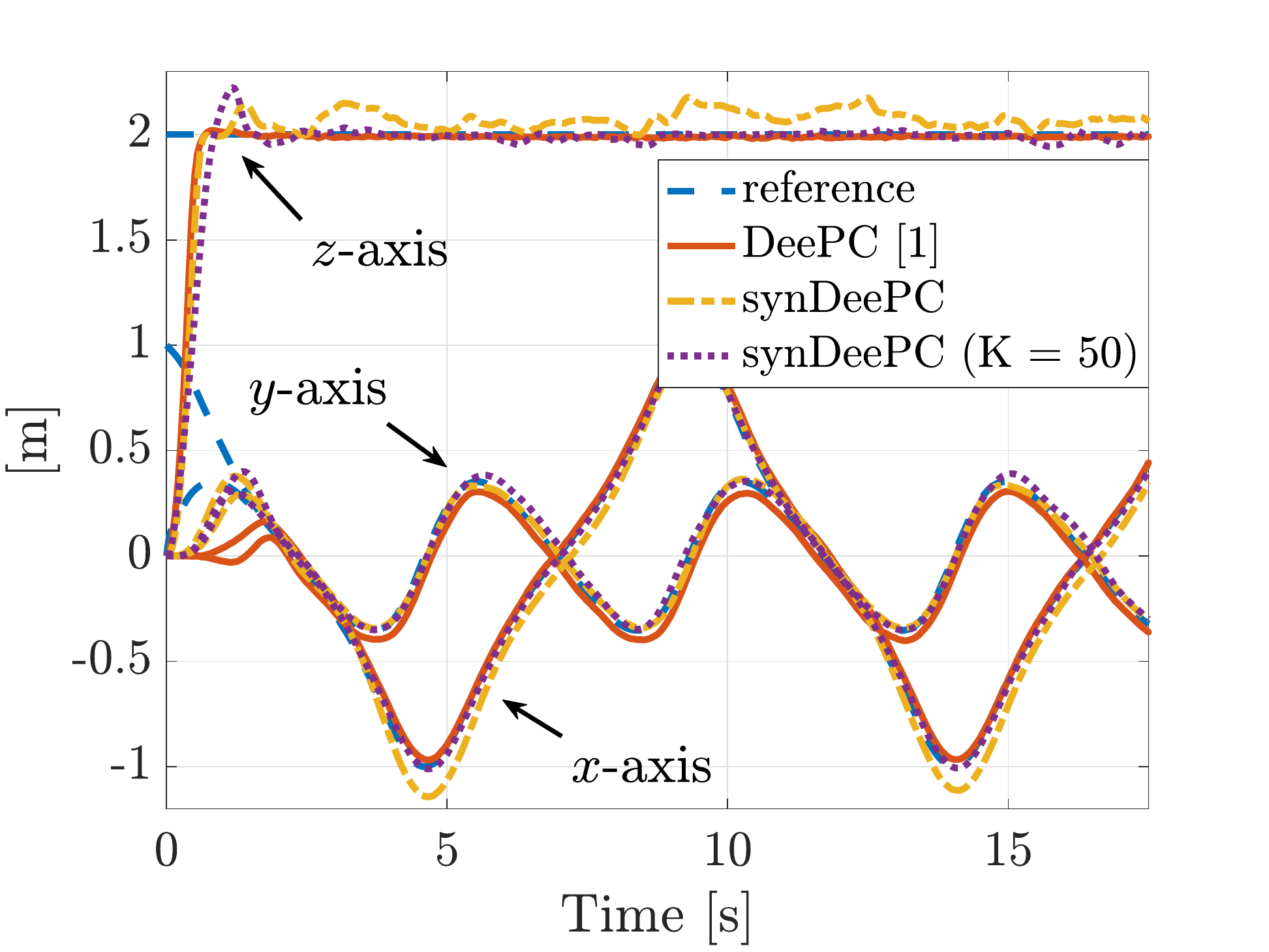}}
		\caption{Dynamical evolution of the controlled quadcopter while following a figure-8 trajectory.}\label{fig:tracking}
	\end{figure}
	\begin{figure}
		\centering
		\includegraphics[width=0.5\columnwidth]{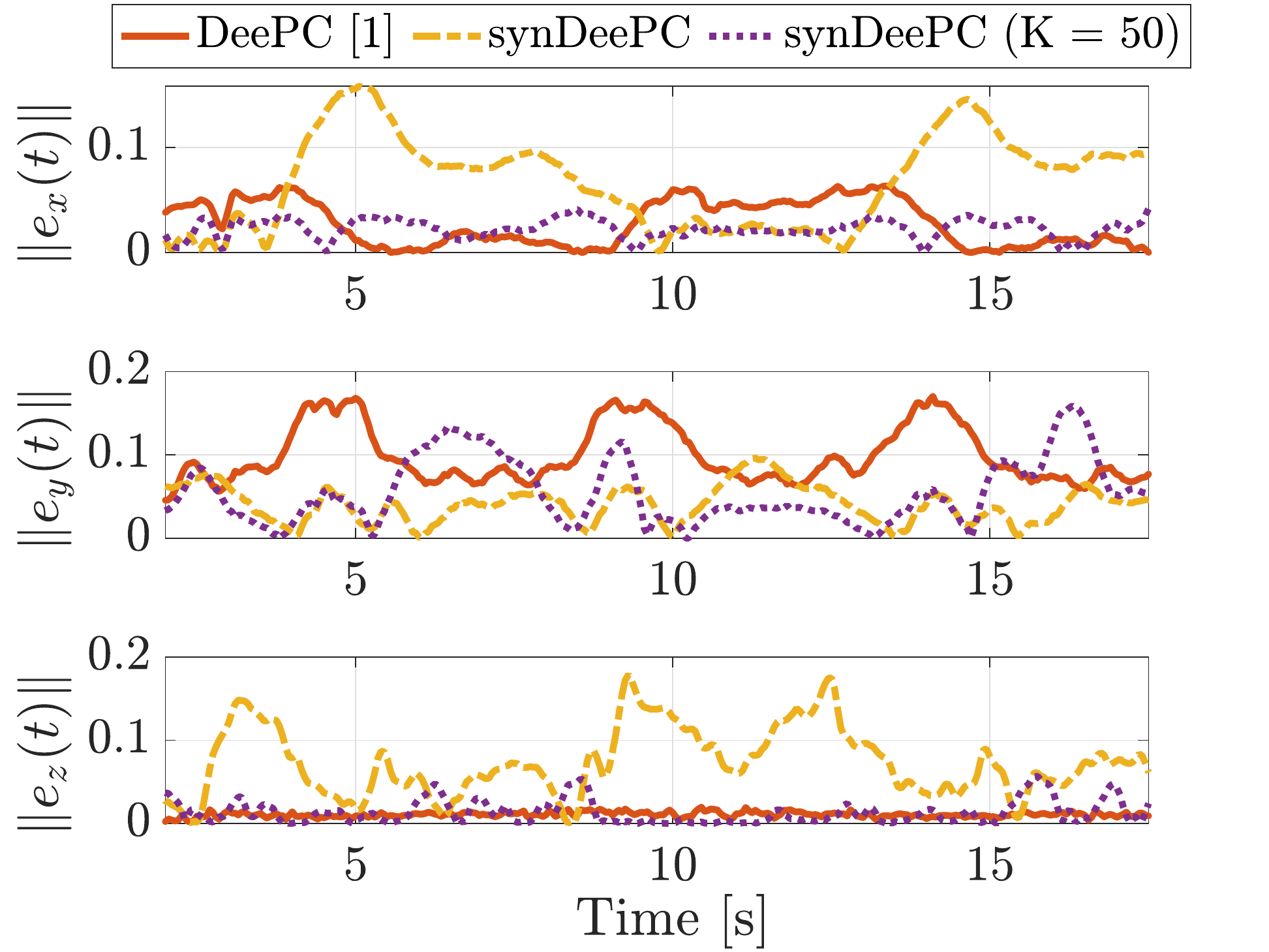}
		\caption{Spatial coordinate tracking errors.}\label{fig:tracking_errors}
	\end{figure}

	\begin{figure}
	\centering
	\includegraphics[width=0.5\columnwidth]{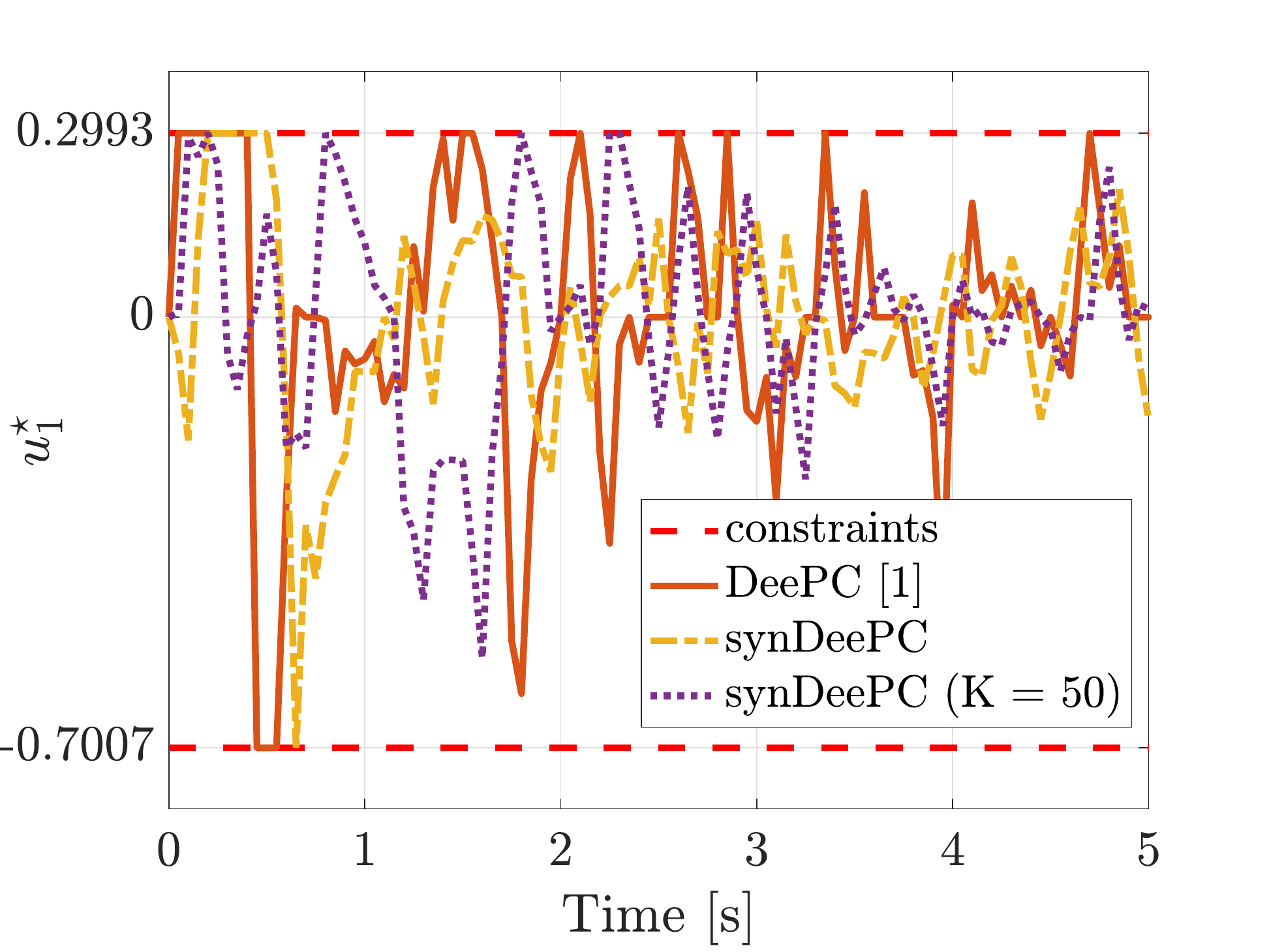}
	\caption{First element of the control input $u^\star$.}\label{fig:control_input}
	\end{figure}

	In Fig.~\ref{fig:tracking} and \ref{fig:tracking_errors}, we compare the trajectory tracking performances of the quadrotor controlled by means of the \gls{DeePC} with full dataset $\mathscr{H}_{31}$ (solid lines) and synthetic dataset, $\mathscr{S}_{31}$, obtained first by reducing to the 50$\%$ the total number of samples, i.e., $S = 92$ (solid-dashed lines), which also correspond to a reduction of the 70$\%$ w.r.t. the total number of samples when considering a longer control horizon, i.e., $K = 50$ (dotted lines). An example of a typical constrained input signal for the rotors can be found in Fig.~\ref{fig:control_input}, where the behaviour of the first element of $u^\star$ in all three cases is shown. Here, we used an equivalent statistic for the noise acting on each measurement channel, i.e., $\nu \sim \mc{N}(0,2^{-7})$, whose value is chosen to match the experimental setup in \cite{elokda2019data}.
	As shown in Fig.~\ref{fig:tracking_errors}, where the position errors of the spatial coordinates are illustrated, the performances of the robust controller computed by means of the synthetic dataset $\mathscr{S}_{31}$ do not degrade markedly compared with the one computed by means of $\mathscr{H}_{31}$, also exhibiting an almost overlapping behaviour when the control horizon $K$ increases ($\mathscr{S}_{51}$). Moreover, from our numerical experience, the step (\texttt{S1}) in Algorithm~\ref{alg:deepc} with a compressed dataset takes approximately $0.64$[s] on average, in sharp contrast to the $1.73$[s] required by the original dataset, see Fig.~\ref{fig:CPU_time_comparison}. On the other hand, a longer control horizon does not lead to a much higher computational time, i.e., $0.85$[s], while considering the whole dataset would take around $3$[s] to solve the \gls{DeePC} optimization problem.

	\begin{figure}
		\centering
		\includegraphics[width=0.5\columnwidth]{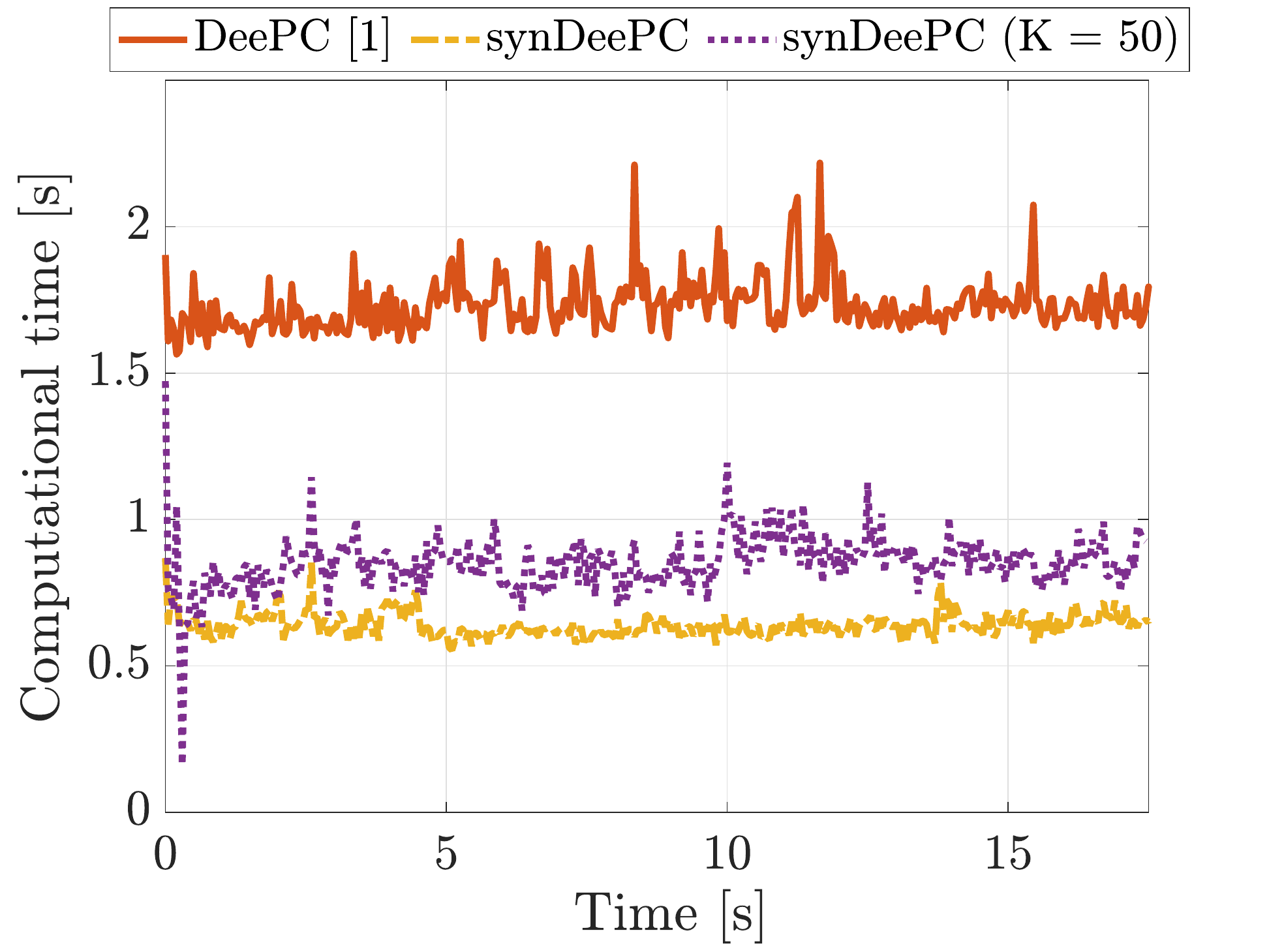}
		\caption{Computational time with Gurobi \cite{gurobi} over the whole trajectory tracking control problem.}\label{fig:CPU_time_comparison}
	\end{figure}

	\begin{figure}
		\centering
		\includegraphics[width=0.5\columnwidth,trim = 0 0 0 0.5cm]{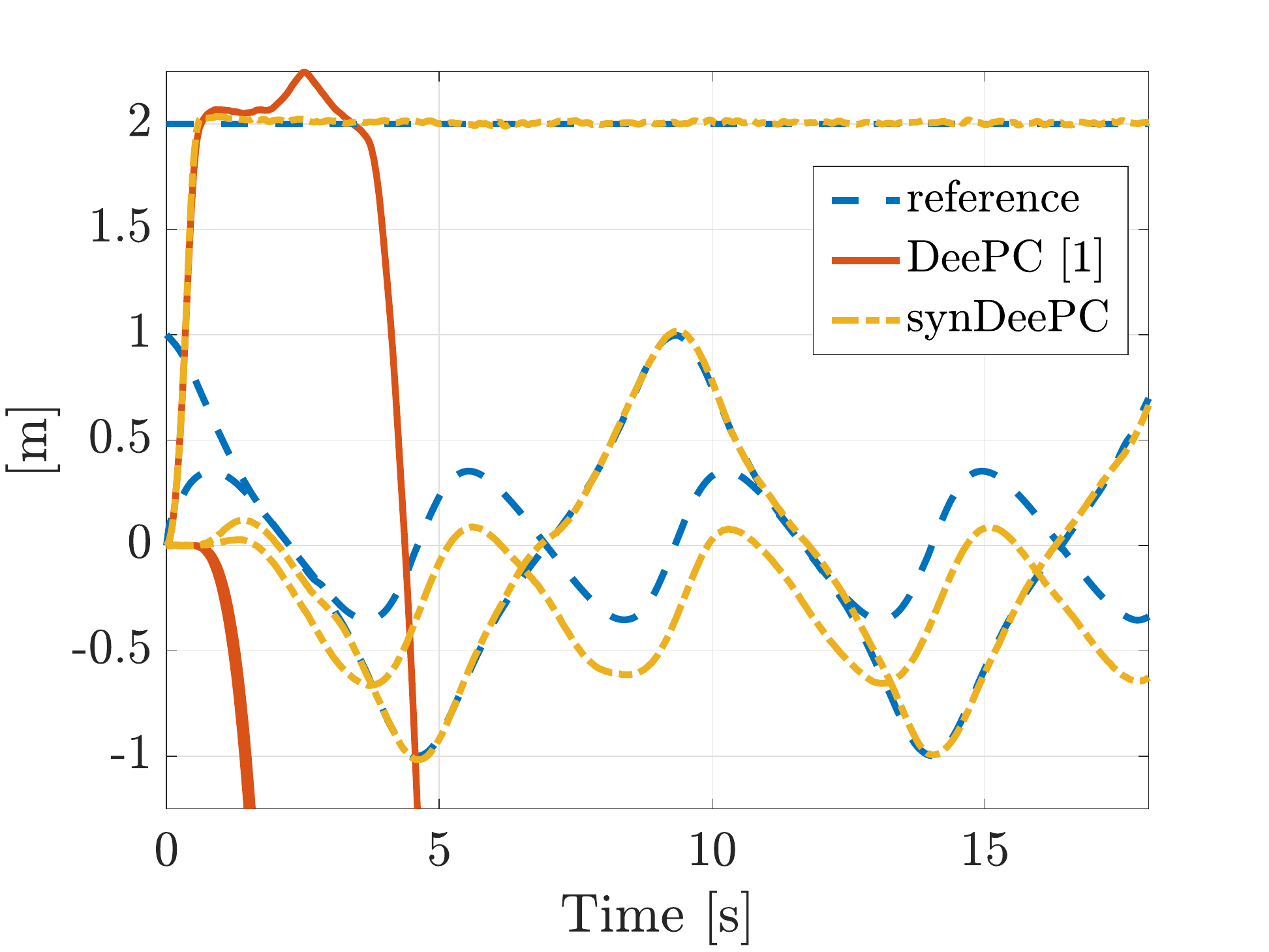}
		\caption{Trajectory tracking performance of the controlled quadcopter with a higher level of noise, $\nu \sim \mc{N}(0,1^{-3})$.}\label{fig:higher_noise}
	\end{figure}

	Finally, we investigate how a higher level of noise acting on the measurements reflects on the control performance. Specifically, we assume an equivalent statistic on the channels, i.e., $\nu \sim \mc{N}(0,1^{-6})$. In this case, with $N = 214$ and main parameters as in Tab.~\ref{tab:sim_param}, we experienced that both controllers do not accomplish the trajectory tracking problem, directly leading to instability. For this reason, we decide to collect $3N$ measurements. As shown in Fig.~\ref{fig:higher_noise}, while the original \gls{DeePC} is not able to follow the figure-8 reference from the very beginning (solid lines), the syn\gls{DeePC} in Algorithm~\ref{alg:deepc} with $S = 306$ introduces an offset on the $y$-axis tracking only, while keeping good performance on both $x$ and $z$ axes (solid-dashed lines).

	\section{Conclusion and Outlook}
	The optimal transport approach promises to be a key tool for the design of synthetic datasets guaranteeing both robust performance and reasonable computational burden for real-time implementation of data-driven controllers. Specifically, we have investigated the benefits of adopting the Wasserstein metric to compress the informative content of a large dataset into a smaller one, also illustrating the performance of the robust controller obtained by means of the synthetic dataset compared to the original one.
	Future research directions will focus on the impact that the discrete measure adopted to compare the empirical distribution associated with the original data, i.e., the vector $\beta$ in \eqref{eq:discrete_transp}, has on the robustness. Intuitively, if one was to allow variable weights in $\beta$, we envision that Fig.~\ref{fig:cost_evol} would be necessarily monotone, assuming one could solve to global optimality. The reason is that the $(S+1)$-th point could always be placed at an arbitrary location with zero mass, and this would require the same Wasserstein distance as with $S$ synthetic points only. It would also mean that $\eta(S) = 0$ for all $S > R$. Moreover, given the input/output structure of the gathered data, we will investigate also the possibility to use different distances to define the Wasserstein metric, as well as a jointly convex reformulation of the whole program in \eqref{eq:wass_optimal}.

	\balance
	\bibliographystyle{IEEEtran}
	\bibliography{ACC_data_driven}

\begin{thebibliography}{10}
\providecommand{\url}[1]{#1}
\csname url@samestyle\endcsname
\providecommand{\newblock}{\relax}
\providecommand{\bibinfo}[2]{#2}
\providecommand{\BIBentrySTDinterwordspacing}{\spaceskip=0pt\relax}
\providecommand{\BIBentryALTinterwordstretchfactor}{4}
\providecommand{\BIBentryALTinterwordspacing}{\spaceskip=\fontdimen2\font plus
\BIBentryALTinterwordstretchfactor\fontdimen3\font minus
  \fontdimen4\font\relax}
\providecommand{\BIBforeignlanguage}[2]{{%
\expandafter\ifx\csname l@#1\endcsname\relax
\typeout{** WARNING: IEEEtran.bst: No hyphenation pattern has been}%
\typeout{** loaded for the language `#1'. Using the pattern for}%
\typeout{** the default language instead.}%
\else
\language=\csname l@#1\endcsname
\fi
#2}}
\providecommand{\BIBdecl}{\relax}
\BIBdecl

\bibitem{willems2005note}
J.~C. Willems, P.~Rapisarda, I.~Markovsky, and B.~L. De~Moor, ``A note on
  persistency of excitation,'' \emph{Systems \& Control Letters}, vol.~54,
  no.~4, pp. 325--329, 2005.

\bibitem{de2019formulas}
C.~De~Persis and P.~Tesi, ``Formulas for data-driven control: Stabilization,
  optimality, and robustness,'' \emph{IEEE Transactions on Automatic Control},
  vol.~65, no.~3, pp. 909--924, 2019.

\bibitem{van2020data}
H.~J. Van~Waarde, J.~Eising, H.~L. Trentelman, and M.~K. Camlibel, ``Data
  informativity: a new perspective on data-driven analysis and control,''
  \emph{IEEE Transactions on Automatic Control}, 2020.

\bibitem{berberich2020robust}
J.~Berberich, A.~Koch, C.~W. Scherer, and F.~Allg{\"o}wer, ``Robust data-driven
  state-feedback design,'' in \emph{2020 American Control Conference
  (ACC)}.\hskip 1em plus 0.5em minus 0.4em\relax IEEE, 2020, pp. 1532--1538.

\bibitem{coulson2019data}
J.~Coulson, J.~Lygeros, and F.~D{\"o}rfler, ``Data-enabled predictive control:
  In the shallows of the {DeePC},'' in \emph{2019 18th European Control
  Conference (ECC)}.\hskip 1em plus 0.5em minus 0.4em\relax IEEE, 2019, pp.
  307--312.

\bibitem{coulson2019regularized}
J.~{Coulson}, J.~{Lygeros}, and F.~{Dörfler}, ``Regularized and
  distributionally robust data-enabled predictive control,'' in \emph{2019 IEEE
  58th Conference on Decision and Control (CDC)}, 2019, pp. 2696--2701.

\bibitem{berberich2020data}
J.~Berberich, J.~K{\"o}hler, M.~A. Muller, and F.~Allgower, ``Data-driven model
  predictive control with stability and robustness guarantees,'' \emph{IEEE
  Transactions on Automatic Control}, 2020.

\bibitem{roweis2000nonlinear}
S.~T. Roweis and L.~K. Saul, ``Nonlinear dimensionality reduction by locally
  linear embedding,'' \emph{Science}, vol. 290, no. 5500, pp. 2323--2326, 2000.

\bibitem{bartlett2006convexity}
P.~L. Bartlett, M.~I. Jordan, and J.~D. McAuliffe, ``Convexity, classification,
  and risk bounds,'' \emph{Journal of the American Statistical Association},
  vol. 101, no. 473, pp. 138--156, 2006.

\bibitem{wang2002new}
J.~Wang and S.~J. Qin, ``A new subspace identification approach based on
  principal component analysis,'' \emph{Journal of process control}, vol.~12,
  no.~8, pp. 841--855, 2002.

\bibitem{wang2006closed}
------, ``Closed-loop subspace identification using the parity space,''
  \emph{Automatica}, vol.~42, no.~2, pp. 315--320, 2006.

\bibitem{scholkopf1997kernel}
B.~Sch{\"o}lkopf, A.~Smola, and K.-R. M{\"u}ller, ``Kernel principal component
  analysis,'' in \emph{International conference on artificial neural
  networks}.\hskip 1em plus 0.5em minus 0.4em\relax Springer, 1997, pp.
  583--588.

\bibitem{van1997closed}
P.~Van~Overschee and B.~De~Moor, ``Closed loop subspace system
  identification,'' in \emph{Proceedings of the 36th IEEE Conference on
  Decision and Control}, vol.~2.\hskip 1em plus 0.5em minus 0.4em\relax IEEE,
  1997, pp. 1848--1853.

\bibitem{mckelvey1996subspace}
T.~McKelvey, H.~Ak{\c{c}}ay, and L.~Ljung, ``Subspace-based multivariable
  system identification from frequency response data,'' \emph{IEEE Transactions
  on Automatic Control}, vol.~41, no.~7, pp. 960--979, 1996.

\bibitem{jansson1998consistency}
M.~Jansson and B.~Wahlberg, ``On consistency of subspace methods for system
  identification,'' \emph{Automatica}, vol.~34, no.~12, pp. 1507--1519, 1998.

\bibitem{villani2003topics}
C.~Villani, \emph{Topics in optimal transportation}.\hskip 1em plus 0.5em minus
  0.4em\relax American Mathematical Society, 2003, no.~58.

\bibitem{peyre2019computational}
G.~Peyr{\'e} and M.~Cuturi, ``Computational optimal transport,''
  \emph{Foundations and Trends{\textregistered} in Machine Learning}, vol.~11,
  no. 5-6, 2019.

\bibitem{mesbah2016stochastic}
A.~Mesbah, ``Stochastic model predictive control: An overview and perspectives
  for future research,'' \emph{IEEE Control Systems Magazine}, vol.~36, no.~6,
  pp. 30--44, 2016.

\bibitem{esfahani2018data}
P.~M. Esfahani and D.~Kuhn, ``Data-driven distributionally robust optimization
  using the {W}asserstein metric: Performance guarantees and tractable
  reformulations,'' \emph{Mathematical Programming}, vol. 171, no. 1-2, pp.
  115--166, 2018.

\bibitem{cuturi2014fast}
M.~Cuturi and A.~Doucet, ``Fast computation of {W}asserstein barycenters,''
  \emph{Journal of Machine Learning Research}, 2014.

\bibitem{ng2000note}
M.~K. Ng, ``A note on constrained k-means algorithms,'' \emph{Pattern
  Recognition}, vol.~33, no.~3, pp. 515--519, 2000.

\bibitem{rolet2016fast}
A.~Rolet, M.~Cuturi, and G.~Peyr{\'e}, ``Fast dictionary learning with a
  smoothed {W}asserstein loss,'' in \emph{Artificial Intelligence and
  Statistics}, 2016, pp. 630--638.

\bibitem{frogner2015learning}
C.~Frogner, C.~Zhang, H.~Mobahi, M.~Araya, and T.~A. Poggio, ``Learning with a
  {W}asserstein loss,'' in \emph{Advances in Neural Information Processing
  Systems}, 2015, pp. 2053--2061.

\bibitem{cuturi2016smoothed}
M.~Cuturi and G.~Peyr{\'e}, ``A smoothed dual approach for variational
  {W}asserstein problems,'' \emph{SIAM Journal on Imaging Sciences}, vol.~9,
  no.~1, pp. 320--343, 2016.

\bibitem{lellmann2014imaging}
J.~Lellmann, D.~A. Lorenz, C.~Schonlieb, and T.~Valkonen, ``Imaging with
  {K}antorovich--{R}ubinstein discrepancy,'' \emph{SIAM Journal on Imaging
  Sciences}, vol.~7, no.~4, pp. 2833--2859, 2014.

\bibitem{cuturi2013sinkhorn}
M.~Cuturi, ``Sinkhorn distances: Lightspeed computation of optimal transport,''
  in \emph{Advances in neural information processing systems}, 2013, pp.
  2292--2300.

\bibitem{bertsekas1989parallel}
D.~P. Bertsekas and J.~N. Tsitsiklis, \emph{Parallel and distributed
  computation: numerical methods}.\hskip 1em plus 0.5em minus 0.4em\relax
  Prentice Hall Englewood Cliffs, NJ, 1989, vol.~23.

\bibitem{beck2013convergence}
A.~Beck and L.~Tetruashvili, ``On the convergence of block coordinate descent
  type methods,'' \emph{SIAM Journal on Optimization}, vol.~23, no.~4, pp.
  2037--2060, 2013.

\bibitem{gurobi}
\BIBentryALTinterwordspacing
L.~Gurobi~Optimization, ``Gurobi optimizer reference manual,'' 2020. [Online].
  Available: \url{http://www.gurobi.com}
\BIBentrySTDinterwordspacing

\bibitem{elokda2019data}
E.~Elokda, J.~Coulson, P.~Beuchat, J.~Lygeros, and F.~D{\"o}rfler,
  ``Data-enabled predictive control for quadcopters,'' 2019, {ETH} Zurich,
  Automatic Control Laboratory.

\end{thebibliography}

\end{document}